\documentclass[journal,comsoc]{IEEEtran}
\usepackage{epsfig,latexsym,amsmath,amssymb,setspace,color,graphicx,algorithm,algorithmic,cases}

\usepackage{amssymb,amsmath,amsfonts,bm,epsfig,graphicx,latexsym}
\usepackage{rotating,setspace,latexsym,epsf,color}
\usepackage{cite,authblk,epstopdf,footnote}
\usepackage{float}
\usepackage{tabularx}
\usepackage{array}
\usepackage{bbm}
\usepackage{amsthm}
\usepackage{thmtools}
\usepackage{algorithm,algorithmic}

\usepackage{eqparbox}
\makeatletter
\declaretheorem[style=plain,qed=$\blacksquare$]{theorem}

\declaretheorem[style=plain,name=Definition,qed=$\blacksquare$]{Definition1}

\declaretheorem[style=plain,name=Remark,qed=$\blacksquare$]{remark}

\allowdisplaybreaks
\IEEEoverridecommandlockouts
\def\mc{\ensuremath\mathcal}
\begin{document}
	%\addtolength{\voffset}{+0.069 in}
	%\sloppy
	%\doublespacing
	
	%% Paper Title
	%% You can use linebreaks \\ within to get better formatting as
	%% desired. 
	\title{Device-to-Device Secure Coded Caching\thanks{This work was supported in part by the National Science Foundation Grants CNS 13-14719 and CCF 17-49665. This work was presented
			in part at Asilomar Conference on Signals, Systems, and Computers, Pacific
			Grove, CA, USA, November 2016.}}
	%\vspace{.2 in}
	\author{Ahmed A. Zewail}
	\author{Aylin Yener}
	\affil{\normalsize Wireless Communications and Networking Laboratory (WCAN)\\
		Electrical Engineering Department\\
		The Pennsylvania State University, University Park, PA 16802.\\
		zewail@psu.edu \qquad yener@engr.psu.edu}
	\maketitle
	\vspace{-0.8in}
%\begin{center}
%\today
%\end{center}

\begin{abstract}
This paper studies device to device (D2D) coded-caching with information theoretic security guarantees. A broadcast network consisting of a server, which has a library of files, and end users equipped with cache memories, is considered. Information theoretic security guarantees for confidentiality are imposed upon the files. The server populates the end user caches, after which D2D communications enable the delivery of the requested files. Accordingly, we require that a user must not have access to files it did not request, i.e., secure caching.  First, a centralized coded caching scheme is provided by jointly optimizing the cache placement and delivery policies. Next, a decentralized coded caching scheme is developed that does not require the knowledge of the number of active users during the caching phase. Both schemes utilize non-perfect secret sharing and one-time pad keying, to guarantee secure caching. Furthermore, the proposed schemes provide secure delivery as a side benefit, i.e., any external entity which overhears the transmitted signals during the delivery phase cannot obtain any information about the database files. The proposed schemes provide the achievable upper bound on the minimum delivery sum rate. Lower bounds on the required transmission sum rate are also derived using cut-set arguments indicating the multiplicative gap between the lower and upper bounds. Numerical results indicate that the gap vanishes with increasing memory size. Overall, the work demonstrates the effectiveness of D2D communications in cache-aided systems even when confidentiality constraints are imposed at the participating nodes and against external eavesdroppers. 
\end{abstract}
%\vspace{-0.1in}
\begin{IEEEkeywords}
Device-to-device communications, coded caching, secure caching, secure delivery, secret sharing.\end{IEEEkeywords}
%\vspace{-.09 in}
\section{Introduction}
Over the past decade, wireless communication systems have transformed from being limited to serving low data rates, e.g., voice calls and text messages, to offering dependable high data rate services, notably, for video content. The demand in massive amounts of data will only increase going forward, leading to potential bottlenecks. Two potential solutions offered towards alleviating network congestion are \textit{device-to-device (D2D) communications} and \textit{caching}. The former shifts some of the traffic load from the core network to the end users, while the later shifts it from the peak to off-peak hours, i.e., to when the network resources is underutilized. More specifically, \textit{D2D communications} utilize the radio channel for end users to communicate directly instead of routing via the network infrastructure \cite{doppler2009device,tehrani2014device,asadi2014survey}, while \textit{caching} stores partial content that may be requested by users in the network in off-peak hours so as to reduce delivery rates to these users during peak-traffic hours\cite{dowdy1982comparative,almeroth1996use}. The seminal reference \cite{maddah2014fundamental} introduced coded caching and demonstrated that, designing the downloading of partial data in off-peak hours, and the delivery signal in peak-hours in a manner to serve multiple users' file demands simultaneously, provides gains that are above and beyond simply placing some partial content in the caches. In particular, it has been shown that jointly designing the cache placement and delivery phases provides a \textit{global caching gain} that results from the ability of serving multiple users by a single transmission, in addition to the \textit{local caching gain} that results from the fact that some of the requested data is locally available in the user's cache. There has been significant recent interest in caching systems, notably in designing coded caching strategies demonstrating gains in various network settings beyond the broadcast network setting of the original reference, see for example, \cite{maddah2016coding,niesen2017coded	,shariatpanahi2015multi,yang2018coded,bidokhti2017benefits,ibrahim2017centralized,ibrahim2017optimization,
zewail2017combination,ji2016fundamental}.

Caching in D2D communications have been pioneered in reference \cite{ji2016fundamental}. In particular, a network where a server, with database of $N$ files, each with size $F$ bits, connected to $K$ users, each equipped with a cache memory of size $MF$ bits, has been considered. In the cache placement phase, the server populates the cache memories of the users with partial content from the server's database. During the delivery phase, in contrast with the communication model in \cite{maddah2014fundamental}, the server remains inactive and the users' requests are to be satisfied via D2D communications only. Both centralized and decentralized schemes were provided. In the centralized schemes, the cache placement and delivery phases are jointly optimized, which requires the knowledge of the number of active users in the system while performing cache placement. In decentralized scheme, this knowledge is not necessary. The fundamental limits of coded caching in device-to-device networks have been further investigated in \cite{ji2016wireless,ji2017fundamental,shabani2016mobility, 
tebbi2017coded,chorppath2017network}. For instance, references \cite{ji2016wireless,ji2017fundamental,shabani2016mobility} have studied the impact of coded caching on throughput scaling laws of D2D networks under the protocol model in \cite{gupta2000capacity}.% Furthermore, reference 

Beside the need of reducing the network load during the peak hours, maintaining secure access and delivery is also essential in several applications, e.g., subscription services. These concerns can be addressed by the so called secure caching and secure delivery requirements studied in server based models. For \textit{secure delivery}  \cite{sengupta2015fundamental,awan2015fundamental,zewail2016coded,zewail2017combination}, any external entity that overhears the signals during the delivery phase must not obtain any information about the database files. 
 In particular, reference \cite{awan2015fundamental} has studied a device-to-device caching system with secure delivery. Utilizing one-time padding, a centralized scheme has been proposed by jointly optimizing the cache placement and delivery phase. The order-optimality of this scheme has been shown in \cite{awan2018bounds}, i.e., the multiplicative gap between the achievable delivery load, in \cite{awan2015fundamental} and the developed lower bound, in \cite{awan2018bounds}, can be bounded by a constant that is independent from the system's parameters. For \textit{secure caching} \cite{ravindrakumar2016fundamental,zewail2016coded,zewail2017combination}, each user should be able to recover its requested file, but must not gain any information about the contents of the files it did not request. %It is worth mentioning that in addition to investigating the simple setups in  \cite{maddah2014fundamental,ji2014fundamental,
%,awan2015fundamental,ravindrakumar2016fundamental}, references  %\cite{karamchandani2014hierarchical,maddah2015cache,shariatpanahi2015multi,
%ji2015caching,ji2015fundamental,naderializadeh2016fundamental,hachem2016layered,
%hachem2016degrees,bidokhti2016erasure,tang2016coded,zewail2016cns,zewail2017coded} have investigated cache-aided systems considering more elaborate network structures. %For instance, references \cite{maddah2015cache,naderializadeh2016fundamental,hachem2016layered,hachem2016degrees} have investigated  wireless interference networks, where the transmitters and receivers are equipped with cache memories, while references \cite{karamchandani2014hierarchical,
%ji2015caching,ji2015fundamental,tang2016coded,zewail2016cns,zewail2017coded} have studied two-hop setups where the network nodes have caching capabilities. In particular, the system operates over two successive phases. In the cache placement phase, the server stores functions of the $N$ files in the users' cache. During the delivery phase, each user requests one of the $N$ files, and advertises its request to all network users. After obtaining the requests of all network users, each user utilizes its cache to transmit a signal that helps the other users to satisfy their requests, while the server remains silent over this phase. Consequently, we reduce the load on the server during the delivery phase, which is a worthy objective \cite{tehrani2014device,doppler2009device,asadi2014survey,ji2016fundamental}.

 In this paper, we investigate the fundamental limits of secure caching in \textit{device-to-device} networks. That is, unlike the settings in \cite{zewail2016coded,ravindrakumar2016fundamental,zewail2017combination}, the server disengages during the delivery phase, and users' requests must be satisfied via D2D communications only. By the end of the delivery phase, each user must be able to reconstruct its requested file, while not being able to obtain any information about the remaining $N\!-\!1$ files. For this D2D model, we derive lower and upper bounds on the rate-memory trade-off. We propose a centralized caching scheme, where the server encodes each file using proper \textit{non-perfect secret sharing schemes} \cite{ito1989secret,ito1987secret,li2015optimal,blakley1984security,blakley1984security}, and generates a set of random keys \cite{shannon1949}. Then, the server carefully places these file shares and keys in the cache memories of the users. During the delivery phase, each user maps the contents of its cache memory into a signal transmitted to the remaining users over a shared multicast link. Next, motivated by the proposed schemes in \cite{jin2016new} under no secrecy requirements, we provide a semi-decentralized scheme, using a grouping-based approach that guarantees secure caching, and does not require the knowledge of the number of active users in the system while populating the users' cache memories. To evaluate the performance of these achievable schemes, we also develop a lower bound on the required transmission sum rate based on cut-set arguments. We show that the multiplicative gap between the lower and upper bounds is bounded by a constant. Furthermore, we observe numerically that this gap vanishes as the memory size increases.

 %, and for systems with realistic parameters, the lower and upper bounds are tight.
%One can expect a decline in the performance compared with the centralized coded caching schemes. Surprisingly, our numerical results shows that even with non-accurate estimates of the expected number of active user in the system, the performance of the decentralized schemes is very close to the one of the centralized schemes.  In this study, we show how to distribute the database content over the cache memories of a set of users, such that no user can gain any information about it from its cache, while all users' requests can be satisfied via exchanging signals between each other. 

By virtue of the D2D model, the delivery load has to be completely transferred from the server to the end users during cache placement in this network, so that no matter what file is demanded by a user, it can be delivered from other users. As such, imposing secure caching requirement will also facilitate secure delivery as we shall see in the sequel. In other words, for the proposed schemes, secure delivery \cite{sengupta2015fundamental,awan2015fundamental,zewail2016coded} is also satisfied as a byproduct.

This work demonstrates that D2D communications can effectively replace a server with full database access despite the fact that each user accesses only a portion of the database and that this is possible with a negligible transmission overhead, while keeping the users ignorant about the database contents. That is to say that, the performance of the system under investigation and the one in \cite{ravindrakumar2016fundamental} are very close to one another for realistic values of the system parameters. We note that while the centralized scheme and its performance were presented in brief in the conference paper \cite{zewail2016fundamental}, the decentralized coded caching scheme and the order-optimality results are presented for the first time in this paper, along with proof details of all results.

The remainder of the paper is organized as follows. In Section \ref{sec:sm}, we describe the system model. In Sections \ref{sec:ach} and \ref{sec:dec}, we detail the centralized and decentralized coded caching schemes, respectively. Section \ref{sec:lower} contains the derivation of the lower bound. In Section \ref{sec:numerical}, we demonstrate the system performance by numerical results. Section \ref{sec:con} summarizes our conclusions. In the following, we will use the notation $[L] \triangleq \{1,\ldots, L\}$, for a positive integer $L$.
%\vspace{-.15 in} 
\section{System Model}\label{sec:sm}
Consider a network where a server, with a database of $N$ files, $W_1,\ldots, W_N$, is connected to $K$ users. The files are assumed to be independent from one another, each has size $F$ bits and is uniformly generated. Each user is equipped with a cache memory with size $MF$ bits, i.e., each user is capable of storing $M$ files. We denote by $M$ the normalized cache memory size and define $Z_k$ to represent the contents of the cache memory at user $k$, where $k\in \{1,2,\ldots,K\}$. The system operates over two consecutive phases, as depicted in Fig. \ref{fig_sm}. 
\begin{figure}[t]
\centering
\includegraphics[scale=0.28]{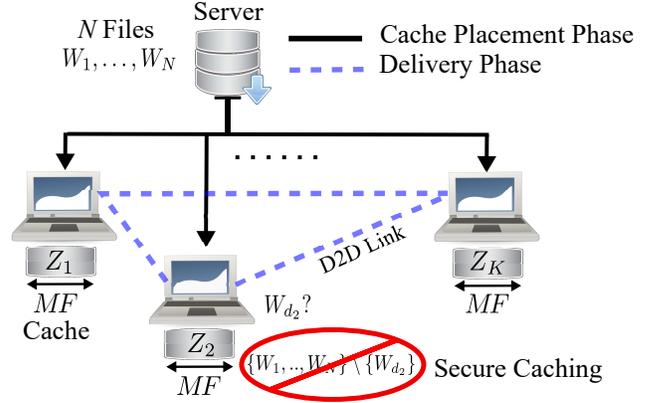}
\centering
\caption{Device-to-device secure coded caching system.}\label{fig_sm}
\end{figure}
\subsection{Cache Placement Phase} 
In this phase, the server allocates functions of its database into the users' cache memories without the knowledge of file demands. These possible allocations are designed to preserve the memory capacity constraint at each user. This is made precise by the following definition.% without the knowledge of its desired file in the near future during the following phase. 
\begin{Definition1}
(Cache Placement): In the cache placement phase, the server maps the files of its database to the cache memories of the users. In particular, the content of the cache memory at user $k$ is given by 
\begin{equation}
Z_k = \phi_k(W_1, W_2,\ldots,W_N), \qquad k =1,2,\ldots,K, 
\end{equation}
where $\phi_k: [2^F]^N\rightarrow [2^F]^M$, such that $H(Z_k)\leq MF$. 
\end{Definition1}
 In this work, aligned with the caching literature, e.g., \cite{sengupta2015fundamental,awan2015fundamental,ravindrakumar2016fundamental,
zewail2016coded,zewail2017combination}, we assume that the cache placement phase is secure, i.e., it is not overheard by any unauthorized entity and the cache contents of each user is not accessible to any other user.
\subsection{Delivery Phase}
During peak traffic, each user requests a file. We assume that the demand distribution is uniform for all users,
 and independent from one user to another \cite{maddah2014fundamental,ji2016fundamental}, i.e., each user can request each file with equal probability. The index of the file requested by user $k$ is $d_k\in[N]$, and $\bm d=(d_1,\ldots,d_K)$ represents the demand vector of all users. % at any request instance. 
Similar to \cite{ji2016fundamental}, we require that the delivery phase is carried our by D2D communications only, i.e., the server participates only in the cache placement phase. Therefore, we need the caches at the users to be able to store the whole library, collectively. Without secrecy requirements, we would need $KM\geq N$ to accomplish this. In Section \ref{sec:ach}, we will see that a larger total memory constraint will be required in order to satisfy the secrecy requirements. With the knowledge of the demand vector $\bm d$, user $k$ maps the contents of its cache memory, $Z_k$, into a signal that is transmitted to all network users over a noiseless interference-free multicast link. From the $K\!-\!1$ received signals and $Z_k$, user $k$ must be able to decode its requested file, $W_{d_k}$, with negligible probability of error. We have the following definition for encoding and decoding at each user. 
\begin{Definition1}
(Coded Delivery): The transmitted signal by user $k$ is given by 
\begin{equation}
X_{k,\bm d} = \psi_k(Z_k,\bm d), 
\end{equation}
where $\psi_k: [2^F]^M\times [N]^K\rightarrow [2^F]^{R_k}$ is the encoding function, $R_k$ is the normalized rate of the transmitted signal by user $k$ and $k \in [K]$. In addition, user $k$ recovers its requested file as
\begin{equation}
\hat W_{d_k} = \mu_k(Z_k,\bm d, X_{1,\bm d},\ldots,X_{k-1,\bm d},X_{k+1,\bm d},\ldots,X_{K,\bm d}),  
\end{equation}   
where $\mu_k: [2^F]^M\times [N]^K\times [2^F]^{\sum_{i\neq k}R_i}\rightarrow [2^F]$ is the decoding function, and $k \in[K]$.
\end{Definition1}
Let $R_T=\sum_{i=1}^K R_i$ be the normalized sum rate of the transmitted signals by all users.% at any request instance.
\subsection{System Requirements}
During the delivery phase, the server remains silent, and all users' requests must be satisfied via D2D communications. Therefore, we have the following reliability requirement.
\begin{Definition1}
(Reliability) For each user to recover its requested file from its received signals and the contents of its cache memory, we need 
\begin{equation}\label{reliableconst}
\max_{\bm d, k\in[K]} Pr(\hat W_{d_k}\neq W_{d_k})\leq \epsilon,
\end{equation}
 for any $\epsilon>0$. 
\end{Definition1}
We impose secure caching constraints on the system. In particular, we require that each user must be able to decode only its requested file, and not be able to obtain any information about the content of the remaining $N-1$ files. 
\begin{Definition1}
(Secure caching) For any $\delta_1>0$, we have
\begin{equation}\label{secrtiveconst}
\max_{\bm d, k\in[K]} I(\bm W_{-d_k}; \bm X_{-k,\bm d},Z_k)\leq \delta_1, 
\end{equation}
where $\bm W_{-d_k}\!\!\!=\!\!\{W_1,\ldots,W_N\} \backslash  \{W_{d_k}\}$, i.e., the set of all files except the one requested by user $k$ and $\bm X_{-k,\bm d}=\{X_{1,\bm d},\ldots,X_{K,\bm d} \}\setminus \{X_{k,\bm d}\}$, i.e., the set of all received signals by user $k$.  
\end{Definition1}

%In this paper, we focus mainly on studying the device-to-device caching system under the secure caching constraint. 
We aim to minimize the sum rate during the delivery phase under reliability and secure caching requirements. Formally, we have the following definition. 
\begin{Definition1}
The secure memory-rate pair $(M, R_T)$ is said to be \textbf{achievable} if $\forall \epsilon, \delta_1 >0$ and $F \rightarrow \infty$, there exists a set of caching functions, $\{\phi_k\}_{k=1}^K$, encoding functions, $\{\psi_k\}_{k=1}^K$, and decoding functions, $\{\mu_k\}_{k=1}^K$, such that (\ref{reliableconst}) and (\ref{secrtiveconst}) are satisfied. The optimal secure memory-rate trade-off is defined as $R_T^{*}=\inf\{R_T: (M,R_T) \mbox{  is achievable} \}.$
%\normalfont
%\begin{equation}
%R_T^{*}=\inf\{R_T: (M,R_T) \mbox{  is securely achievable} \}.
%\end{equation}
\end{Definition1}
We are also interested in the secure delivery requirement, defined below.  
\begin{Definition1}
(Secure Delivery) Any eavesdropper that overhears the transmitted signals during the delivery phase must not obtain any information about the contents of the data phase files. Therefore, we have 
\begin{align}\label{securedelivery}
\max_{\bm d} I(W_1,\ldots,W_N; X_{1,\bm d},\ldots,X_{K,\bm d})\leq \delta_2,  
\end{align}
for any $\delta_2>0$.
\end{Definition1}

\begin{remark}
We will see that for the D2D setting we consider, our proposed schemes for secure caching will automatically satisfy the secure delivery requirement.
\end{remark}

\begin{remark}
In general, secure delivery and secure caching requirements do not have to imply one another. For instance, if $M\geq N$, secure delivery is trivially satisfied by storing the entire database at each user during the cache placement phase violating the secure caching requirements. An example for the reverse scenario, i.e., where secure caching does not imply the secure delivery can be found in subsection \ref{specialcase}.%s, as from (\ref{secrtiveconst}), it is clear that each user must not be able to recover any file from the contents of its cache memory only. %Thus, the transmission rate is strictly positive even for the large values of $M$.   
\end{remark} 
 We aim to minimize the total delivery load, i.e., the total transmission rate, by designing the cache contents and the delivery strategy while maintaining the secure caching requirement. Since no user should be able to decode a file he did not request, and the requests of the users will be only revealed at the beginning of the delivery phase, we must design the cache's contents of each user in a way that does not reveal any information about the system files. This makes the cache placement problem relevant to the problem of multiple assignment in secret sharing \cite{ito1987secret,ito1989secret,li2015optimal}, in the sense that, we aim to distribute the library over the set of  end users such that the shares assigned to each of them cannot reveal any information about the files. Here the size of shares allocated to each user must not exceed its cache storage capacity, $MF$. Additionally, by the end of the delivery phase, each user must be able to decode only its requested file. With the D2D model, we require the system to maintain self-sustainability without the participation of the server during the delivery phase. Thus, the caches' contents at all users must be able to regenerate the library files. In the following two sections, we provide two schemes that minimize the delivery load while maintaining the systems' requirements.    
\section{Centralized Coded Caching Scheme}\label{sec:ach}
In this section, we consider a scenario where the server is able to perform cache placement in a centralized manner. That is, the server knows the total number of users in the system, $K$, at the beginning of the cache placement phase. We utilize non-perfect secret sharing schemes \cite{yamamoto1986secret,blakley1984security,secretsharing} to encode the database files. The basic idea of the non-perfect secret sharing schemes is to encode the secret in such a way that accessing a subset of shares does not reduce the uncertainty about the secret, and only accessing all shares does. For instance, if the secret is encoded into the scaling coefficient of a line equation, the knowledge of one point on the line does not reveal any information about the secret as there remain infinite number of possibilities to describe the line. One can learn the secret only if two points on the line are provided, and then can do so precisely. We will utilize \textit{non-perfect secret sharing schemes}, formally defined as follows.
\begin{Definition1} \cite{yamamoto1986secret,blakley1984security,secretsharing} For a file $W$ with size $F$ bits, an $(m,n)$ non-perfect secret sharing scheme generates $n$ shares, $S_1,S_2,\ldots,S_n$, such that accessing any $m$ shares does not reveal any information about the file $W$, i.e., \begin{equation}
I(W;\mathcal{S})=0, \quad \forall \mathcal{S}\subseteq\{S_1,S_2,\ldots,S_n\}, |\mathcal{S}|\leq m.
\end{equation}
Furthermore, we have 
\begin{equation}
H(W|S_1,S_2,\ldots,S_n)=0,
\end{equation}
i.e., the file $W$ can be reconstructed from the $n$ shares.
\end{Definition1}
For large enough $F$, an $(m,n)$ non-perfect secret sharing scheme exists with shares of size equals to $\frac{F}{n-m}$ bits \cite{yamamoto1986secret,blakley1984security,secretsharing}. We chose to use schemes from this class as they give shares with sizes equal to the secret size divided by the gap, $\frac{F}{n-m}$ bits. By contrast, perfect secret sharing schemes \cite{shamir1979share} give shares of size equal to the secret size, $F$ bits. Therefore, the non-perfect secret sharing schemes are more efficient in our case in terms of storage and delivery load.
\subsection{Cache Placement Phase}\label{gencache}
First, we present a scheme that works for $M=\frac{Nt}{K-t}+\frac{1}{t}+1$, and $t\in[K-1]$, noting that the remaining values of $M$ can be achieved by memory sharing \cite{maddah2014fundamental}. That is, for any value of $M$, we pick the most two adjacent values, $M_1$ and $M_2$, such that $M_i=\frac{Nt}{K-t_i}+\frac{1}{t_i}+1$, $i=1,2$, $t_i\in[K-1]$, and $M_1\leq M \leq M_2$. We determine the sharing parameter $\alpha\in [0,1]$, by solving the equation $M=\alpha M_1+ (1-\alpha) M_2$. Then, each file, $W_n$, is divided into two subfiles, $W_n^1$ and $W_n^2$, of sizes $\alpha F$ and $(1-\alpha) F$ bits, respectively. The achievability scheme is obtained by applying the scheme designed for the system with memory $M_i$ on the subfiles $W_n^i$, and $i=1,2$.

 As a first step, the server encodes each file in the database using a non-perfect secret sharing scheme \cite{yamamoto1986secret,blakley1984security,secretsharing}. In particular, a file, $W_n$, is encoded using a $\left(t {{K-1}\choose {t-1}},t{{K}\choose {t}}\right)$ non-perfect secret sharing scheme. Each share, with size $F_s$ bits, is denoted $S_{n,\mathcal{T}}^j$, where $j=1,\ldots,t$, and $\mathcal{T}\subseteq [K]$ with $|\mathcal{T}|=t$, and 
\begin{equation}
F_s=\frac{F}{t {{K}\choose {t}}-t {{K-1}\choose {t-1}}}=\frac{F}{(K-t){{K-1}\choose {t-1}}}.
\end{equation}
 We refer to the set $\mathcal{T}$ by the \textit{allocation set} as it determines how the shares will be allocated in the users' caches. In particular, the server places the shares $S_{n,\mathcal{T}}^j$, $\forall j, n$ in the cache of user $k$ whenever $k\in \mathcal{T}$. Also, the  parameter $t$ can be see as number of users that will store the same share.

Additionally, the server generates $(t+1){{K}\choose{t+1}}$ independent keys, i.e., they are independent from one another and independent from the library files. In particular, each key is uniformly distributed over $[2^{F_s}]$, and is denoted by $K_{\mathcal{T}_K}^i$, where $i=1,\ldots,t+1$, and $\mathcal{T}_K\subseteq [K]$ with $|\mathcal{T}_K|=t+1$. The server places the keys $K_{\mathcal{T}_K}^i$, $\forall i$, in user $k$'s cache if $k\in \mathcal{T}_K$, i.e., $\mathcal{T}_K$ represents the key allocation set. Therefore, the cached content by user $k$ at the end of the cache placement phase is given by 
\begin{equation}\label{cenZ}
Z_k=\left\{S_{n,\mathcal{T}}^i,K_{\mathcal{T}_K}^j: k\in  \mathcal{T},\mathcal{T}_K, \mbox{ and } \forall i,n,j  \right\}.
\end{equation}
We summarize the cache placement procedure in Algorithm \ref{alg_place1}. In the following remark, we verify that this placement satisfies the cache memory capacity constraint. 
\begin{remark}
In the aforementioned placement scheme, each user stores $t{{K-1}\choose{t-1}}$ shares of each file and $(t+1){{K-1}\choose{t}}$ distinct keys, thus the accumulated number of cached bits is given by
\begin{align}\label{memacc}
Nt&{{K-1}\choose{t-1}}F_s+(t+1){{K\!-\!1}\choose{t}}F_s=\frac{Nt}{K-t}F+(1\!+\!\frac{1}{t})F\!=MF.
\end{align}
It follows that we have 
%\small
\begin{align}
%t&=\frac{1+(M-1)K+\sqrt{\left(1+(M-1)K\right)^2-4K(N+M-1)}}{2(N+M-1)}\nonumber \\&
t=\frac{1+(M-1)K+\sqrt{\left(1-(M-1)K\right)^2-4KN}}{2(N+M-1)}.
\end{align}
Clearly, the proposed allocation scheme satisfies the cache memory capacity constraint at each user.
\end{remark}
\begin{algorithm}[t]
\begin{algorithmic}[1]
\REQUIRE $\{ W_{1}, \dots, W_{N}\}$ %$a_{\mc S}, \mc S \in \mc A_{\mc S} $
\ENSURE $ Z_k, k \in [K]$
\FOR{$l \in [N]$}
\STATE Encode $W_l$ using an $(t {{K-1}\choose {t-1}},t{{K}\choose {t}})$ non-perfect secret sharing scheme
$\rightarrow S_{l,\mathcal{T}}^{j}$, $j=1,\ldots,t$, and $\mathcal{T}\subseteq [K]$ with $|\mathcal{T}|=t$. 
 \ENDFOR
\FOR{$\mathcal{T}_K\subseteq [K]$ with $|\mathcal{T}_K|=t+1$}
\STATE Generate independent keys $K_{\mathcal{T}_K}^i$, $i=1,\ldots,t+1$.
 \ENDFOR
\FOR{$k \in [K]$}
\STATE $Z_k \leftarrow \left\{K_{\mathcal{T}_K}^i:  k\in \mathcal{T}_K, \forall i \right\} \bigcup \bigcup_{l \in [N]} \left\{S_{l,\mathcal{T}}^j:  k\in \mathcal{T}, \forall j \right\}$
\ENDFOR
\end{algorithmic}
 \caption{Cache placement procedure}\label{alg_place1}
\end{algorithm}
\begin{remark}\label{minmem}
We note that the minimum value of the normalized cache size, $M$, that is need to apply the proposed scheme is $M_{\text{min}}=M\vert_{t=1}=2+\frac{N}{K-1}$. For a system without secrecy requirements \cite{ji2016fundamental}, we need $M\geq \frac{N}{K}$, while with secure delivery, the scheme in \cite{awan2015fundamental} requires $M\geq 2+\frac{N-2}{K}$. It is evident that, with secrecy requirements, more memory is required, as the users not only cache from the data but also cache the secure keys.      
\end{remark} 
\vspace{-0.25 in}
\subsection{Coded Delivery Phase}\label{cendel}
At the beginning of the delivery phase, each user requests one of the $N$ files and the demand vector is known to all network users. To derive an upper bound on the required transmission sum rate, we focus our attention on the worst case scenario. We concentrate on the more relevant scenario of $N\geq K$.

The delivery procedure consists of ${K \choose {t+1}}$ transmission instances. At each transmission instance, we consider a set of users $\mathcal{S}\subseteq[K]$, where $|\mathcal{S}|=t+1$. We refer to $\mathcal S$ as the transmission set. For $k\in \mathcal{S}$, user $k$ multicasts the following signal of length $F_s$ bits
\begin{equation}
X_{k,\bm d}^{\mathcal{S}}=\oplus_{l \in \mathcal{S}\setminus\{k\} } S^j_{d_l,\mathcal{S} \setminus\{l\}}\oplus K^i_{\mathcal{S}}.
\end{equation}
Note that the index $i$ is chosen such that each key is used only once,  while the index $j$ is chosen to ensure that each transmission is formed by shares that had not been transmitted in previous transmissions by the other users in $\mathcal{S}$. For example, they can be chosen as the relative order of the user's index with respect to the indices of the remaining users in $\mathcal{S}$. Thus, in total, the transmitted signal by user $k$ can be expressed as
\begin{equation}
X_{k,\bm d}=\bigcup_{\mathcal{S}:\ \! k\in\mathcal{S}, \ \! \mathcal{S}\subseteq[K],|\mathcal{S}|=t+1 }\{X_{k,\bm d}^{\mathcal{S}}\}.
\end{equation} 
Observe that the cache memories of the users from any subset, $\mathcal{S}_t\subset\mathcal{S}$, with $|\mathcal{ S}_t|=t$ contain $t$ shares of the file requested by the user in $\mathcal{S}\setminus \mathcal{ S}_t$, as can be seen from (\ref{cenZ}). Thus, utilizing its cache contents, each user in $\mathcal{S}$ obtains $t$ shares from its requested file during this instance of transmission, i.e., the user in $\mathcal{S}\setminus \mathcal{ S}_t$ obtains the shares $S_{d_{\mathcal{S}\setminus \mathcal{S}_t},\mathcal{S}_t}^j \ \forall j$. 

Observe also that user $k$ belongs to ${K-1} \choose {t}$ different choices of such subsets of the users, thus at the end of the delivery phase, user $k$ obtains $t{{K-1} \choose {t}}$ \textit{new} shares of its requested files, in addition to the cached $t{{K-1} \choose {t-1}}$ shares. Therefore, user $k$ can decode its requested file from its $t{{K} \choose {t}}$ shares, i.e., the reliability requirement (\ref{reliableconst}) is satisfied. Delivery procedure is summarized in Algorithm \ref{alg_delv1}. %The following remark notes that the proposed scheme satisfies (\ref{secrtiveconst}) and (\ref{securedelivery}).

 \begin{algorithm}[t]
\begin{algorithmic}[1]
\REQUIRE $\bm d$
\ENSURE $X_{k,\bm d}, k\in [K] $
\FOR{$k\in [K]$}
\FOR{$\mc S  \in [K], |\mc S|=t+1, k\in \mc S$} 
\STATE  $X_{k,\bm d}^{\mathcal{S}}\leftarrow  \oplus_{l \in \mathcal{S}\setminus\{k\} }   S^j_{d_l,\mathcal{S} \setminus\{l\}}\oplus K^i_{\mathcal{S}}  $, for some choice of $i$ and $j$
\ENDFOR
\STATE $X_{k,\bm d}\leftarrow \bigcup_{\mathcal{S}\subseteq[\hat K], k \in \mc S} \{X_{k,\bm d}^{\mathcal{S}}\}$
\ENDFOR
\end{algorithmic}
 \caption{Delivery procedure}\label{alg_delv1}
\end{algorithm} 

\subsection{Rate Calculation}
Now, we focus our attention on calculating the required transmission rate. Note that there are $K \choose {t+1}$ different choices of the set $\mathcal{S}$. For each choice, $t+1$ signals of length $F_s$ bits are transmitted, thus the total number of the transmitted bits is given by
\begin{equation}
R_TF=(t+1){K \choose {t+1}} F_s=\frac{K}{t}F.
\end{equation} 
Consequently, we can achieve the following normalized sum rate
%\small  
%\begin{equation}
%R_T\leq\frac{2K(N+M-1)}{1+(M-1)K+\sqrt{\left(1+(M+1)K\right)^2-4K(N+M-1)}}.
%\end{equation}
\begin{equation}
R_T=\frac{2K(N+M-1)}{1+(M-1)K+\sqrt{\left(1-(M+1)K\right)^2-4KN}}.
\end{equation}
%\vspace{-.07 in}
Therefore, we can state the following theorem.% which defines an upper bound on the centralized secure coded caching sum rate.
\begin{theorem}\label{them:cen}
Under centralized placement, for $M=\frac{Nt}{K-t}+\frac{1}{t}+1$, and $t\in[K-1]$, the secure sum transmission rate is upper bounded by
\begin{equation}
R_T^*\leq R_T^C\leq\!\frac{2K(N+M\!-\!1)}{1+(M\!-\!1)K+\sqrt{\left(1\!-\!(M\!-\!1)K\right)^2\!-\!4KN}}.
\end{equation}
In addition, using memory sharing \cite{maddah2014fundamental}, we can achieve the convex envelope of the points given by the values $M\!=\!\frac{Nt}{K-t}\!+\!\frac{1}{t}\!+\!1$, and $t\in[K-1]$.
\end{theorem}
 %This concludes the proof of Theorem $1$.the convex envelope of the above points, defined for each $M$, is also achievable.

\subsection{Secrecy Analysis}   
For user $k$, the cache's contents, $Z_k$ given by (\ref{cenZ}), contains only $t {{K-1}\choose {t-1}}$ shares, from each file, resulting from a $\left(t {{K-1}\choose {t-1}},t{{K}\choose {t}}\right)$ non-perfect secret sharing scheme. Therefore, $Z_k$, by itself, cannot reveal any information about the files to user $k$. 
 
During the delivery phase, if at any instance, user $k$ belongs to the transmission set, $\mathcal{S}$, then the transmitted signals are formed from the shares of  the requested file by user $k$, $W_{d_k}$, and shares that have been already placed in the cache of user $k$ during the cache placement phase, i.e., from $Z_k$. When user $k$ does not belong to the transmission set, all the transmitted signals are encrypted using one-time pads, unknown to user $k$, thus, user $k$ cannot gain any information from these signals \cite{shannon1949}. Therefore, the secure caching constraint, (\ref{secrtiveconst}), is satisfied.
	
We observe that the server has generated $(t+1){{K} \choose {t+1}}$ independent keys with lengths equal to the share size. Thus, with a proper selection of the encrypting key for each transmission, we can ensure a unique use of each key, i.e., one-time padding. The above discussion implies that  the secrecy of the transmitted signals, from any external wiretapper that accesses the network links during the delivery phase, is also guaranteed \cite{shannon1949}. One-time pads are essential to ensure the secure caching requirement in (\ref{secrtiveconst}), whereas the secure delivery requirement in (\ref{securedelivery}), is satisfied as a byproduct.

\subsection{An Illustrative Example}
%We demonstrate the aforementioned scheme by the following example.
 Consider a system with four users and a library consists of four files, $W_1,W_2,W_3,W_4$, i.e., $K=N=4$. Each user has a normalized memory size $M=\frac{11}{2}$, which gives us $t=2$ and indicates that each of the resulting shares will be cached by two different users. The server encodes each file using $(6,12)$ non-perfect secret sharing scheme. For a file, $W_n$, the server generates $12$ shares, which we label by $S_{n,\mathcal{T}}^i$ where $i=1,2$, $\mathcal{T}\subset\{1,2,3,4\}$ and $|\mathcal{T}|=2$, each of size $F/6$ bits.
 
  Furthermore, the server generates the set of keys $K_{\mathcal{T}_K}^j$, uniformly distributed over $\{1,\ldots,2^{F/6}\}$, where $j=1,2,3$, $\mathcal{T}_K\subset\{1,2,3,4\}$ and $|\mathcal{T}_K|=3$.

User $k$ stores the shares $S_{n,\mathcal{T}}^i$, and the keys $K_{\mathcal{T}_K}^j$ whenever $k\in\mathcal{T}$ and $k\in\mathcal{T}_K$, $\forall n,j,i$, respectively. Therefore, the cache contents at the users are given by 
\begin{equation}\nonumber
Z_{1}=\begin{Bmatrix} 
S_{n,12}^i,S_{n,13}^i,S_{n,14}^i, \ \forall n,i, K_{123}^j,K_{124}^j,K_{134}^j, \ \forall j
 \end{Bmatrix},
\end{equation}
\begin{equation}\nonumber
Z_{2}=\begin{Bmatrix} 
S_{n,12}^i,S_{n,23}^i,S_{n,24}^i, \forall n,i,
K_{123}^j,K_{124}^j,K_{234}^j, \ \forall j
\end{Bmatrix},
\end{equation}
\begin{equation}\nonumber
Z_{3}=\begin{Bmatrix} 
S_{n,13}^i,S_{n,23}^i,S_{n,34}^i, \forall n,i,
K_{123}^j,K_{134}^j,K_{234}^j, \ \forall j
\end{Bmatrix},
\end{equation}
\begin{equation}\nonumber
Z_{4}=\begin{Bmatrix} 
S_{n,14}^i,S_{n,24}^i,S_{n,34}^i, \forall n,i,
K_{124}^j,K_{134}^j,K_{234}^j, \ \forall j
\end{Bmatrix}.
\end{equation}  
Each user caches $6$ shares of each file. We observe that the caches will not be able to reveal any information about the unrequested files thanks to the non-perfect secret sharing encoding. Also, note that the cache capacity constraints at all the users are satisfied.

 Now, consider the delivery phase, where user $k$ requests the file $W_k$, i.e., $\bm d=(1,2,3,4)$. In this case, the users transmit the following signals.
\begin{align}\nonumber
X_{1,\bm d}=\begin{Bmatrix}  
&S_{2,13}^1 \oplus  S_{3,12}^1   \oplus K_{123}^1,  S_{4,13}^1 \oplus  S_{3,14}^1  \oplus  K_{134}^1, \nonumber \\
& \quad S_{2,14}^1 \oplus  S_{4,12}^1  \oplus K_{124}^1
 \end{Bmatrix},
\end{align}
\begin{align}\nonumber
X_{2,\bm d}=\begin{Bmatrix} 
&S_{1,23}^2\oplus S_{3,12}^2 \oplus K_{123}^2,  S_{4,23}^2\oplus S_{3,24}^2 \oplus K_{234}^2, \nonumber \\
& \quad S_{1,24}^2\oplus S_{4,12}^2 \oplus K_{124}^2
 \end{Bmatrix},
\end{align}
\begin{align}\nonumber 
X_{3,\bm d}=\begin{Bmatrix}  
& S_{1,23}^1\oplus S_{2,13}^2 \oplus K_{123}^3,  S_{4,13}^2\oplus S_{1,34}^1 \oplus K_{134}^2,
 \nonumber \\
&S_{2,34}^1\oplus S_{4,23}^1 \oplus K_{234}^3
 \end{Bmatrix},
\end{align}
\begin{align}\nonumber
X_{4,\bm d}=\begin{Bmatrix}  
&S_{1,24}^1\oplus S_{2,14}^2 \oplus K_{124}^3,  S_{1,34}^2\oplus S_{3,14}^2 \oplus K_{134}^3,
 \nonumber \\
&S_{2,34}^2\oplus S_{3,24}^1 \oplus K_{234}^1
 \end{Bmatrix}.
\end{align}
From its received signals, $X_{2,\bm d}$, $X_{3,\bm d}$ and $X_{4,\bm d}$, and utilizing its cached content, user $1$ gets $S_{1,23}^1$, $S_{1,23}^2$, $S_{1,24}^1$, $S_{1,24}^2$, $S_{1,34}^1$ and $S_{1,34}^2$. Thus, user $1$ can reconstruct its requested file, $W_1$, from its $12$ shares. Similarly, users $2$, $3$ and $4$ are able to decode files $W_2$, $W_3$ and $W_4$, respectively.

We observe that user $k$ will only obtain new shares of its requested file $W_k$, thus it cannot gain any information about the remaining files, $\{W_1,W_2,W_3,W_4\}\setminus \{W_k\}$. This is done by proper selection of the keys so that each user cannot gain any information about the remaining three files. In addition, each signal is encrypted using one-time pad which ensures the secrecy of the database files from any external eavesdropper as in \cite{awan2015fundamental}. In this delivery procedure, each user participates by $3$ distinct transmissions, each of size $F/6$ bits, thus $R_T^C=2$. Comparing with the system in \cite{ravindrakumar2016fundamental}, where the server is responsible for the delivery phase, we see that a normalized secure rate $\simeq 1.3$ is achievable, for the same system parameters. This difference is due to limited access of the shares at each user, unlike the case in \cite{ravindrakumar2016fundamental} where the server can access all shares during the delivery phase, i.e., the cost of having D2D delivery.

\subsection{Secure Caching without Secure Delivery for $M=N(K-1)$}\label{specialcase}
The scheme described above ensures that the requirements in (\ref{secrtiveconst}) (and (\ref{securedelivery})) are satisfied. The encryption keys are essential to achieve both. In the following, we study a special case where we can provide a scheme that achieves secure caching, i.e., satisfy (\ref{secrtiveconst}), without the necessity of satisfying the secure delivery constraint, i.e., (\ref{securedelivery}).  
More specifically, when $M=N(K-1)$, we can achieve a normalized rate equals to $\frac{K}{K-1}$ without utilizing encryption keys. In particular, each file is encoded using $((K-1)^2,K(K-1))$ non-perfect secret sharing scheme. The resulting shares, each of size $F_s=\frac{F}{K-1}$ bits, are indexed by $S_{n,i}^j$, where $n$ is the file index, $j=1,\ldots,K-1$, and $i=1,\ldots,K$. The server allocates the shares $S_{n,i}^j$,  $\forall j,n$ and $i\neq k$ in the memory of user $k$, i.e., 
\begin{equation}\label{censpecial}
Z_k=\{S_{n,i}^j:  \forall j,n  \mbox{ and } i\neq k\}.
\end{equation}
Thus, each user stores $N(K-1)^2$ shares, which satisfies the memory capacity constraint. 

At the beginning of the delivery phase, each user announces its request. Again, we assume that the users request different files. User $k$ multicasts the following signal to all other users 
\begin{equation}
X_{k,\bm d}=\oplus_{l \in [K]\setminus\{k\} } S^j_{d_l,l},
\end{equation}   
where $j$ is chosen to ensure that each transmission is formed by fresh shares which had not been included in the previous transmissions. From its received $K-1$ signals, user $k$ can extract the shares $S^j_{d_k,k}$, $\forall j$. By combining these shares with the ones in its memory, user $k$ recovers its requested file, $W_{d_k}$. The total number of bits transmitted under this scheme is $R_TF=KF_s.$ 
%\begin{equation}
%R_TF=KF_s. 
%\end{equation}
Thus, the following normalized sum rate, under the secure caching constraint (\ref{secrtiveconst}), is achievable for $M=N(K-1)$,
\begin{equation}
R_T=\frac{K}{K-1}. 
\end{equation}
This rate matches the cut set bound as in Section \ref{sec:lower}.
\subsection{Discussion}
The above scheme, in subsection \ref{specialcase}, satisfies only the secure caching constraint (\ref{secrtiveconst}), without ensuring the protection from any external eavesdropper that overhears the transmitted signals during the delivery phase. On the other hand, the general scheme, presented in subsections \ref{gencache} and \ref{cendel}, achieves the same rate, i.e., $R_T^C=\frac{K}{K-1}$, when $M=N(K-1)+\frac{K}{K-1}$, while satisfying the secure caching constraint (\ref{secrtiveconst}) and the secure delivery constraint (\ref{securedelivery}), simultaneously. In other words, an additional memory at each user with size $\frac{K}{K-1}F$ bits is required to ensure the additional requirement of secure delivery.

We observe that the encryption keys serve to satisfy both the secure caching and secure delivery requirements. Therefore, one can think about the satisfaction of the secure delivery requirement as a byproduct of the general scheme in subsections \ref{gencache} and \ref{cendel}, i.e., the secure delivery comes for free while satisfying the secure caching constraint, whenever $M \leq \frac{N(K-2)}{2}+\frac{1}{K-2}+1$. 
Under different network topologies, secure delivery may require additional cost. For example, in recent reference \cite{zewail2017combination}, we have shown that there is no need to use encryption keys to satisfy the secure caching requirements in the setting of combination networks. This is possible due to the unicast nature of the network links, which is not the case in the system under investigation, as we assume that the users communicate with each other via multicast links.

%  We present a decentralized sequential coded caching scheme that guarantees a practical self-sustainability of the system, i.e., the users are capable of satisfying their requests without the participation of the server during the delivery phase.
\section{Decentralized Coded Caching Scheme}\label{sec:dec}
In this section, we provide a decentralized coded caching scheme, \cite{maddah2015decentralized}, for our setup. The proposed scheme is motivated by the ones in \cite{jin2016new} for multicast coded caching setup without secrecy requirements \cite{maddah2014fundamental} \cite{maddah2015decentralized}. It does not require the knowledge of the number of active users of the delivery phase during cache placement. This scheme operates over two phases as follows.
\vspace{-.2 in}
\subsection{Cache Placement Phase}\label{dec:place}
The main idea of the cache placement scheme is to design the cache contents for a number of users $L$ that is less than the number of users in the system during the delivery phase, i.e., $K$. $L$ is in effect a lower bound on the expected number of active users in the system. 

For a given $L$ and $M=\frac{Nt}{L-t}+\frac{2}{t}+1$, and $t\in[L-1]$, each file in the database is encoded using a suitable non-perfect secret sharing scheme. In particular, a file, $W_n$, is encoded using $\left(t {{L-1}\choose {t-1}},t{{L}\choose {t}}\right)$ non-perfect secret sharing scheme. We obtain $t{{L}\choose {t}}$ shares, each with size $\bar F_s$, where
\begin{equation}
\bar F_s=\frac{F}{t {{L}\choose {t}}-t {{L-1}\choose {t-1}}}=\frac{F}{(L-t){{L-1}\choose {t-1}}}.
\end{equation}
Each share is denoted by $S_{n,\mathcal{T}}^j$, where $n$ is the file index, i.e., $n\in[N]$, $j=1,\ldots,t$, and $\mathcal{T}\subseteq [L]$ with $|\mathcal{T}|=t$. 
The server prepares the following set of cache contents, $\bar{Z_l}$, 
\begin{equation}
\bar{Z_l}=\{S_{n,\mathcal{T}}^j: l\in \mathcal{T}, \ \ \forall j, n\}, \qquad \qquad l=1,2,\ldots,L.
\end{equation}
Once user $k$ joins the system, it caches the content $\bar Z_{l_k}$ where $l_k=k \mod L$. Such allocation results in dividing the set of active users into $\lceil\frac{K}{L}\rceil$ virtual groups. In particular, we group the first $L$ users to join the system in group $1$, and the users from $L+1$ to $2L$ in group $2$ and so on. Note that each group from $1$ to $\lceil\frac{K}{L}\rceil-1$ contains $L$ users, and the group $\lceil\frac{K}{L}\rceil$ contains $K-(\lceil\frac{K}{L}\rceil-1)L$ users. These groups are formed sequentially in time.

As explained in Section \ref{sec:ach}, we require the server to generate a set of random keys to be shared between the users. For group $u$, $u=1,\ldots,\lceil\frac{K}{L}\rceil-1$, the server generates the keys $K_{u,\mathcal{T}_K}^i$, where $i=1,\ldots,t+1$, $\mathcal{T}_K\subseteq [L]$ and $|\mathcal{T}_K|=t+1$. Each key is uniformly distributed over $[2^{\bar F_s}]$. User $l_k$ from group $u$ stores the keys $K_{u,\mathcal{T}_K}^i$, $\forall i$, whenever $l_k\in \mathcal{T}_K$. 

In addition, the server generates the keys $K_{u^*,\mathcal{T}_K}^i$, where $i=1,\ldots,t+1$, and $\mathcal{T}_K\subseteq [L], |\mathcal{T}_K|=t+1$, and allocates these keys in the cache memories of the users in groups $1$ and $\lceil\frac{K}{L}\rceil$ as follows. The keys $\{K_{u^*,\mathcal{T}_K}^i, \ \forall i\}$ are cached by user $l_k$ from group $\lceil\frac{K}{L}\rceil$, as long as $l_k\in \mathcal{T}_K$. User $l_k$ from group $1$ stores the keys $K_{u^*,\mathcal{T}_K}^j$ for only one specific $j$ whenever $l_k\in \mathcal{T}_K$. This index $j$ is chosen such that the users from group $1$ store different keys.  

In summary, at the end of cache placement, cache contents of user $k$ are given by 
\begin{equation}\label{decenall}
Z_k=\begin{cases} &\{\bar Z_{l_k}, K_{1,\mathcal{T}_K}^i,K_{u^*,\mathcal{T}_K}^j: l_k\in \mathcal{T}_K, \forall i, \mbox{ for a specific } j\}, \\
 & \qquad \qquad \qquad \qquad \mbox{if } 1 \leq k \leq L,\\
&\{ \bar Z_{l_k}, K_{u,\mathcal{T}_K}^i: u=\lceil\frac{k}{L}\rceil, l_k\in \mathcal{T}_K, \forall i\}, \\
 & \qquad \qquad \qquad \qquad \mbox{if } L\!+\!1 \leq k \leq K\!-\!(\lceil\frac{K}{L}\rceil\!-\!1)L,\\
&\{ \bar Z_{l_k}, K_{u^*,\mathcal{T}_K}^i: l_k\in \mathcal{T}_K, \forall i\}, \\
 & \qquad\qquad \qquad \qquad \mbox{if } K\!-\!(\lceil\frac{K}{L}\rceil\!-\!1)L+\!1 \leq k \leq K.
\end{cases}
\end{equation}
 
\begin{remark}
We need to ensure that this allocation procedure does not violate the memory capacity constraint at each user. Observe that each user stores the same amount of the encoded file shares, however, the users from group $1$ stores more keys than the other users. Thus, satisfying the memory constraint at the users in group $1$ implies satisfying the memory constraint at all network users.  
Each user in group $1$ stores $Nt{{L-1}\choose{t-1}}$ shares and $(t+2){{L-1}\choose{t}}$ keys. Thus, the total number of the stored bits is given by
\begin{equation}\label{memacc2}
Nt{{L\!-\!1}\choose{t\!-\!1}}\bar F_s\!+\!(t\!+\!2){{L\!-\!1}\choose{t}}\bar F_s\!=\!\frac{Nt}{L\!-\!t}F\!+\!(1\!+\!\frac{2}{t})F\!=\!MF,
\end{equation}
and from (\ref{memacc2}), we get
%\small
\begin{align}
%t&=\frac{2+(M-1)L+\sqrt{\left(2+(M-1)L\right)^2-8L(N+M-1)}}{2(N+M-1)}\nonumber \\
t=\frac{2+(M-1)L+\sqrt{\left(2-(M-1)L\right)^2-8LN}}{2(N+M-1)}.
\end{align}
Therefore, the proposed scheme satisfies the cache capacity constraint at each user.
\end{remark}
\subsection{Coded Delivery Phase}\label{dec:delivery}
We focus our attention on the worst case demand, where $K$ users request $K$ different files. The delivery phase is divided into $\lceil\frac{K}{L}\rceil$ stages. At each stage, we focus on serving the users of one group. For any stage $u$, where $u=1,\ldots,\lceil\frac{K}{L}\rceil-1$, the delivery process during stage $u$ is performed in a way similar to the one described in subsection \ref{cendel} with $K=L$ to serve the requests of users in group $u$. In particular, at each transmission instance, we consider $\mathcal{S}\subseteq [L]$, where $|\mathcal{S}|=t+1$. User $k$, with $l_k\in \mathcal{S}$, multicasts a signal, of length $\bar F_s$ bits, given by
\begin{equation}
K^i_{u,\mathcal{S}} \oplus_{l_v \in \mathcal{S}\setminus\{l_k\} } S^j_{{d}_v,\mathcal{S} \setminus\{l_v\}},
\end{equation}
where the index $i$ is chosen in way that guarantees the uniqueness of the key utilized for each transmission. 
From the cache placement phase, we observe that any $t$ users belong to the set $\mathcal{S}$ share $t$ shares of the file requested by the remaining user that is in $\mathcal{S}$. Thus, each user in $\mathcal{S}$ obtains $t$ shares from its requested file during this instance of transmission. At the end of stage $u$, each user from group $u$ can decode its requested file from its $t{L \choose {t}}$ shares.

Since, there are $L \choose {t+1}$ different choices of the set $\mathcal S$, and for each choice $t+1$ signals of length $\bar F_s$ are transmitted, the total number of the transmitted bits to serve the users from group $u$ is
\begin{equation}
R_uF=(t+1){L \choose {t+1}} \bar F_s=\frac{L}{t}F, \qquad u=1,\ldots,\lceil\frac{K}{L}\rceil-1.
\end{equation} 
Now, we focus on serving the users of the last group, i.e., group $\lceil\frac{K}{L}\rceil$. First, recall that the number of users in this group is $p\triangleq K-(\lceil\frac{K}{L}\rceil-1)L< L$, thus these users cannot satisfy their requests via device-to-device communications between them only. We require some of the users from group $1$ to participate in this last stage of the delivery phase. In particular, the users indexed by $l_k$, with $\l_k=p+1,\ldots,L$, from group $1$ forms a virtual group with the users from group $\lceil\frac{K}{L}\rceil$, such that the resulting group contains $L$ users. 
Note that, at this stage, the requests of the users from group $1$ have been already served, during stage $1$. Therefore, at each transmission instance, we consider only the sets $\mathcal{S}\subseteq[L]$, where $|\mathcal{S}|=t+1$ with $l_k\in \mathcal{S} $ and $l_k \in[p]$. We define the sets $\mathcal S_{u^*}$ and $\mathcal S_{u^*}^c$ to represent the subset of $\mathcal{S}$ that contains the users from group $\lceil\frac{K}{L}\rceil$ and group $1$, respectively, i.e., $\mathcal S_{u^*} \cup \mathcal S_{u^*}^c=\mathcal{S}$. Since, we only care now about serving the users in group $\lceil\frac{K}{L}\rceil$, we neglect any set $\mathcal S$ with $\mathcal S_{u^*}=\{\}$. For the sets that contain only one user, user $k$, from group $\lceil\frac{K}{L}\rceil$, i.e., $l_k \in \mathcal S_{u^*}$, $|\mathcal S_{u^*}|=1$, and $|\mathcal S_{u^*}^c|=t$, each user in the set $\mathcal S_{u^*}^c$ transmits 
\begin{equation}
K^j_{{u^*},\mathcal{S}} \oplus S^i_{{d}_{k},\mathcal{S} \setminus\{l_k\}},
\end{equation}
where $i$ is chosen to ensure that from every transmission user $k$ obtains a different share from its requested file.\newline 
For the sets that contain more than one user from group $\lceil\frac{K}{L}\rceil$, i.e., $|\mathcal S_{u^*}| \geq 2$, each user in the set $\mathcal{S}$ multicasts a signal of length $\bar F_s$ given by
\begin{equation}
K^j_{u^*,\mathcal{S}} \oplus_{l_v \in \mathcal S_{u^*}\setminus\{l_k\} } S^i_{{d}_v,\mathcal S_{u^*} \setminus\{l_v\}}.
\end{equation}
By taking into account all possible sets with $|\mathcal S_{u^*}|\geq 1$, the total number of the transmitted bits during this stage is given by
\begin{equation}
R_{u^*}F=pt {{L-p} \choose {t}} \bar F_s+\sum_{u=2}^{\min(p,t)} (t+1){{L-p} \choose {t-u+1}} \bar F_s.
\end{equation} 
%Therefore, $R_{u^*}$ can be expressed as
%\begin{equation}
%R_{u^*}=\frac{nt {{L-n} \choose {t}}+\sum_{u=2}^{n} (t+1){{L-n} \choose {t-u+1}}}{(L-t){{L-1} \choose {t-1}}}.
%\end{equation}
Consequently, we can obtain the following upper bound on the normalized sum rate
%\small  
\begin{equation}
R_T^{D}= R_{u^*}+\left(\lceil\frac{K}{L}\rceil\!-\!1\right)R_u.
\end{equation}
\begin{theorem}
For any integer $L\leq K$, $M=\frac{Nt}{L-t}+\frac{2}{t}+1$, and $t\in[L-1]$, the secure sum rate under decentralized coded caching is upper bounded by
\begin{align}
R_T^*\leq R_T^D&\leq\frac{2L(N+M-1)\left(\lceil\frac{K}{L}\rceil-1\right)}{2\!+\!(M\!-\!1)L\!+\!\sqrt{\left(2\!-\!(M\!-\!1)L\right)^2\!\!-\!8LN}}\nonumber \\
& \qquad+\frac{pt \langle{{L-p} \choose {t}} \rangle+\sum_{u=2}^{\min(p,t)} (t+1)\langle{{L-p} \choose {t-u+1}}\rangle}{(L-t){{L-1} \choose {t-1}}}, 
\end{align}
where $\langle{{h} \choose {r}}\rangle={{h} \choose {r}}$ whenever $h\geq r$ and $0$ otherwise, and $p=K-(\lceil\frac{K}{L}\rceil-1)L$. In addition, the convex envelope of the above points, defined for each $M$, is also achievable.
\end{theorem}
Using memory sharing \cite{maddah2014fundamental}, we can achieve the convex envelope of the points given by the values $M=\frac{Nt}{L-t}+\frac{2}{t}+1$, and $t\in[L-1]$.% This concludes the proof of Theorem $2$.

\subsection{Discussion}
In the decentralized coded caching scheme proposed in \cite{sengupta2015fundamental}, for server-based coded caching with secure delivery, key placement is done in a centralized manner after a decentralized caching of a fraction $\frac{M-1}{N-1}$ of each file, without the knowledge of the users' demands. For server-based systems with secure caching, a decentralized scheme was proposed in \cite{ravindrakumar2016fundamental}. 

We note that developing decentralized schemes for D2D coded caching systems is more involved compared with decentralized schemes for server-based coded caching systems \cite{maddah2014fundamental}. This due to the requirement that in D2D the server must disengage the delivery process, i.e., the end users collectively must possess pieces of the entire library. When there are no secrecy requirements, reference \cite{ji2016fundamental} has proposed a decentralized D2D coded caching scheme, which utilizes maximum distance separable (MDS) codes to encode the files at the server to satisfy the users' requests without the participation of the server during the delivery phase.

For our D2D secure coded caching, we utilize a grouping-based approach that allows disengaging the server from the delivery phase, and the key placement is done during the cache placement phase, without the knowledge of the users' demands. We choose this grouping-based approach instead of utilizing the MDS encoding in \cite{ji2016fundamental}, as  each user not only needs to store the keys used in encrypting its intended signals but also the keys that are used to encrypt its transmitted signals. Therefore, applying a decentralized cache placement based on MDS coding requires the users to dedicate a large fraction of their cache memories for the keys to be allocated by the server during the delivery phase after announcing the demand vector. By contrast, our proposed scheme ensures a practical self-sustainable system with reasonable fraction of each cache memory dedicated to encryption keys. Once a group of $L$ users joins the system, sequentially, the server can place the keys in the memories of these users. At the end of the cache placement phase, before the beginning of the delivery phase, the server allocates the keys to be used in encrypting the signals intended to the last group in the caches of the users of the first and the last group. We remark that this grouping-based approach can be used to develop new decentralized schemes for multicast coded caching scenarios with secure delivery and secure caching that were considered in \cite{sengupta2015fundamental,ravindrakumar2016fundamental,awan2015fundamental} as well.

 We observe, from (\ref{decenall}), that the cache memory at each user is divided into two partitions, one for the shares of the files, and the other one for the keys. The partition assigned for the shares $\bar{Z_l}$ can be identical in multiple users. Thus, encrypting the signal with a one-time pad is necessary to satisfy the secure caching requirement. Note that each user from group $1$ which participates in the last stage of the delivery phase, knows only the keys that it will use to encrypt its transmitted signal, thus it cannot gain any new information from the signals transmitted during this last stage. The secure delivery requirement is satisfied as a byproduct as in the centralized scheme of Section \ref{sec:ach}. %Our proposed scheme is the first decentralized scheme that is proposed for device-to-device coded caching systems with secrecy requirements, i.e., secure delivery or secure caching.

\subsubsection{The Choice of $L$}
 A key element in designing the aforementioned semi-decentralized scheme, is the choice of the parameter $L$, which can be determined by observing the number of users in the system during the peak traffic hours over a sufficient amount of time. Then, we can choose $L$ as the minimum value among the observed numbers of users in the system.
  We note that as long as $L$ is close to $K$, the exact number of active users in the system, we can benefit from more multicast opportunities which helps in reducing the overall delivery load.   
	
We note that a minor potential drawback of the provided scheme is that some users cache memories may be under-utilized by a small fraction. In particular, other than the users from group 1 who participate in serving the last group, a very small fraction of size, $\frac{1}{t}$, that is not scaled with the library size, is not utilized from each user memory. This fraction can be see as a cost for disengaging the server from the delivery phase. This fraction cannot be used to cache from data directly due to the secure caching requirement. A good estimate of $L$, i.e., choosing $L$ close to $K$, will reduce the number of users that do not fully utilize their memory.

\subsubsection{User Mobility During the Placement Phase}	
 If a user, $f$, leaves the system during the cache placement phase, then its cached contents, $Z_f$, should be assigned by the server to populate the cache memory of the first user to join the system after this departure. If no user joins the system before the beginning of the delivery phase, then the server can update the contents of the last user that joined the system with $Z_f$.
\section{Lower Bound}\label{sec:lower}
In this section, we derive a lower (converse) bound on the normalized sum of the required sum rate. The derivation is based on cut-set arguments \cite{cover2006elements}, similar to \cite{ji2016fundamental,sengupta2015beyond}. 

Assume that the first $s$ users, where $s\in\{1,2,\ldots,\min(N/2,K)\}$, request the files from $1$ to $s$, such that user $i$ requests $W_i$, $i\in\{1,2,\ldots,s\}$. The remaining users are assumed to be given their requested files by a genie. We define $\bm X_1$ to represent the transmitted signals by the users to respond to these requests, i.e., $\bm X_1=\{X_{1,(1,\ldots,s)},\ldots,X_{K,(1,\ldots,s)}\}$.  At the next request instance, the first $s$ users request the files from $s+1$ to $2s$, such that user $i$ requests $W_{s+i}$. These requests are served by transmitting the signals $\bm X_2=\{X_{1,(s+1,\ldots,2s)},\ldots,X_{K,(s+1,\ldots,2s)}\}$. We proceed in the same manner, such that at the request instance $q$, the first $s$ users request the files from $(q-1)s+1$ to $qs$, such that user $i$ requests $W_{(q-1)s+i}$, and the users transmit the signals $\bm X_q=\{X_{1,((q-1)s+1,\ldots,qs)},\ldots,X_{K,((q-1)s+1,\ldots,qs)}\}$, where $q\in\{1,\ldots,\lfloor N/s \rfloor\}$. In addition, we define $\bm{ \bar X_q}=\{X_{s+1,((q-1)s+1,\ldots,qs)},\ldots,X_{K,((q-1)s+1,\ldots,qs)}\}$ to denote the set of the transmitted signals by the users indexed by $s+1$ to $K$ at request instance $q$.

 From the received signals over the request instances $1,2,\ldots,\lfloor N/s \rfloor$ and the information stored in its cache, i.e., $Z_i$, user $i$ must be able to decode the files $W_i, W_{i+s},\ldots,W_{i+(\lfloor N/s \rfloor-1)s}$. Consider the set of files $ \mathcal{ \bar W}=\{W_1,\ldots,W_{(q-1)s+k-1},W_{(q-1)s+k+1},\ldots,W_{s\lfloor N/s \rfloor} \}$, i.e., the set of all requested files excluding the file, $W_{(q-1)s+k}$, which was requested by user $k$ at the request instance $q$. Therefore, we have
\begin{align}
(s&\lfloor N/s\rfloor-1)F=H(\mathcal{ \bar W}) \nonumber \\
&\leq H(\mathcal{ \bar W})-H(\mathcal{ \bar W}|\bm X_1,\ldots,\bm X_{\lfloor N/s \rfloor}, Z_1,\ldots,Z_s)+\epsilon\label{lowerreliable}\\
&=I(\mathcal{ \bar W};\bm{ \bar X_1},\ldots,\bm  {\bar X_{\lfloor N/s \rfloor}}, Z_1,\ldots,Z_s)+\epsilon\\
&=I(\mathcal{ \bar W};\bm{ \bar X_q}, Z_k)+I(\mathcal{ \bar W};\bm{ \bar X_1},\ldots,\bm{ \bar X_{q-1}},\bm{ \bar X_{q+1}},\ldots,\bm{ \bar X_{\lfloor N/s \rfloor}}, \nonumber \\
& \qquad \qquad \quad \qquad Z_1,\ldots,Z_{k-1},Z_{k+1},\ldots,Z_s|\bm{ \bar X_q}, Z_k)+\epsilon.\label{chain}
\end{align}
Step (\ref{lowerreliable}) follows from (\ref{reliableconst}) as the users must be able to decode their requested files utilizing their cache's contents and received signals. To simplify the notation, we define 
\begin{align}
& \bm{\mathcal{X}}=\{\bm{\bar X_1},\ldots,\bm{\bar X_{q-1}},\bm{\bar X_{q+1}},\ldots,\bm{ \bar X_{\lfloor N/s \rfloor}}\} \nonumber \\ &\qquad\mbox{ and } \qquad \bm{\mathcal Z} =\{ Z_1,\ldots,Z_{k-1},Z_{k+1},\ldots,Z_s\}\nonumber.  
\end{align}
Now, (\ref{chain}) can be expressed as
\begin{align}
I(&\mathcal{ \bar W};\bm{ \bar X_q}, Z_k)+I(\mathcal{ \bar W};\bm{\mathcal{X}},\bm{\mathcal{Z}}|\bm{ \bar X_q}, Z_k)+\epsilon\nonumber \\& \leq I(\mathcal{ \bar W};\bm{\mathcal{X}},\bm{\mathcal{Z}}|\bm{ \bar X_q}, Z_k)+\epsilon+\delta \label{lowerconf}\\
&=H(\bm{\mathcal{X}},\bm{\mathcal{Z}}|\bm{ \bar X_q}, Z_k)-H(\bm{\mathcal{X}},\bm{\mathcal{Z}}|\mathcal{ \bar W},\bm{ \bar X_q}, Z_k)+\epsilon+\delta\\
&\leq H(\bm{\mathcal{X}},\bm{\mathcal{Z}}|\bm{ \bar X_q}, Z_k)+\epsilon+\delta\\ 
&\leq H(\bm{\mathcal{X}},\bm{\mathcal{Z}})+\epsilon+\delta\\
&=H(\bm{\mathcal{X}})+H(\bm{\mathcal{Z}}|\bm{\mathcal{X}})+\epsilon+\delta\\
&\leq H(\bm{\mathcal{X}})+H(\bm{\mathcal{Z}})+\epsilon+\delta\\
&\leq \sum_{j=1,j\neq q}^{\lfloor N/s \rfloor} H(\bm {\bar X_j})+\sum_{i=1,i\neq k}^s H(Z_i)+\epsilon+\delta\\
&\leq (\lfloor N/s \rfloor-1)\frac{K-s}{K}RF+(s-1)MF+\epsilon+\delta.
\end{align}
Note that step (\ref{lowerconf}) is due to (\ref{secrtiveconst}). Therefore, we can get
\begin{equation}
R_T\geq \frac{K[(s\lfloor N/s\rfloor-1)-(s-1)M]}{(\lfloor N/s \rfloor-1)(K-s)}.
\end{equation}
Taking into account all possible cuts, we obtain the lower bound stated in the following theorem. 
\begin{theorem}
 The achievable secure rate is lower bounded by
\begin{equation}\label{d2dlower}
R_T^{*}\geq \max_{s\in\{1,2,\ldots,\min(K,N/2)\}}\frac{K[(s\lfloor N/s\rfloor-1)-(s-1)M]}{(\lfloor N/s \rfloor-1)(K-s)}.
\end{equation}
\end{theorem}

 	\begin{remark}
We note that $R_T^{*}\geq \frac{K}{K-1}$, which is obtained by setting $s=1$ in Theorem 3, can be achieved whenever $M\geq N(K-1)$, using the proposed scheme in subsection \ref{specialcase}.
\end{remark} 
In addition, the multiplicative gap between the upper bound in Theorem \ref{them:cen} and the above lower bound is bounded by a constant as stated in the following theorem. 
\begin{theorem}\label{them:gap}
For $M\geq 2 +\frac{N}{K-1}$, there exists a constant, $c$, independent of all the system parameters, such that 
\begin{equation}
 1\leq\frac{R_{T}^{C}}{R_{T}^*}\leq c. 
  \end{equation}
\end{theorem}
\begin{proof}
See the Appendix.
\end{proof}
\section{Numerical Results}\label{sec:numerical}
\begin{figure}[t]
\includegraphics[width=3.2 in,height=2.2 in]{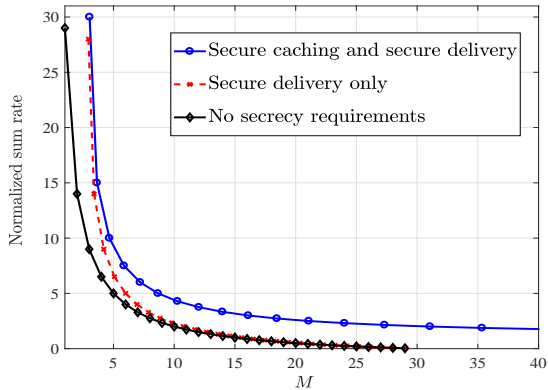}
\centering
\caption{\small Comparison between the required transmission rates under different system requirements for $N=K=30$.}\label{fig_diff_requirements}
\end{figure}
In this section, we demonstrate the performance of the proposed schemes numerically. Fig. \ref{fig_diff_requirements} shows the performance of D2D coded caching systems under different requirements. In particular, we compare our system that provides both secure caching and secure delivery with the system with just secure delivery \cite{awan2015fundamental} and the one with no secrecy constraints \cite{ji2016fundamental}. For the latter two cases, the rate is equal to zero wherever $M\geq N$, as the entire database can be stored in each cache memory. However, by setting $s=1$ in the lower bound stated in Theorem $3$, we get that the sum rate under secure caching is bounded below by $\frac{K}{K-1}$. 
\begin{figure}[t]
\includegraphics[width=3.2 in,height=2.2 in]{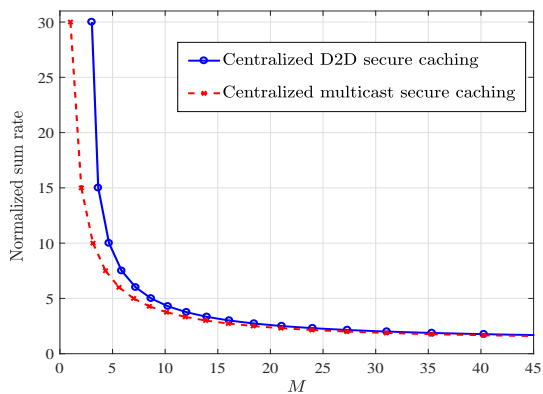}
\centering
\caption{\small The achievable secure rates for the single server and D2D coded caching for $N=K=30$.}\label{fig_single_vs_D2D}
\end{figure}

In Fig. \ref{fig_single_vs_D2D}, we compare the performance of our system and the one considered in \cite{ravindrakumar2016fundamental}. As expected, the system, in \cite{ravindrakumar2016fundamental}, where the server, with full access to the file shares, is responsible for the delivery phase, achieves lower transmission rate compared with the considered setup where the delivery phase has to be performed by users, each of which has limited access to the files shares. Interestingly, we observe that the gap between the required transmission rates vanishes as $M$ increases, i.e., the loss due to accessing a limited number of shares at each user is negligible when $M$ is sufficiently large. 

Fig. \ref{fig_uppervslower} shows that the gap between the lower and upper bounds decreases as $M$ increases. As mentioned before, Theorem $3$ points out that the sum rate is bounded below by $\frac{K}{K-1}$, by setting $s=1$. For large enough $M$, our proposed schemes achieves a sum rate equal to $\frac{K}{K-1}$, which matches the lower bound. 

In Fig. \ref{fig_decen_vs_cen}, we plot the achievable rates under different choices of $L$ in a system with $K=100$. It is worth noting that even with $L=60$, which is much smaller than the  number of users in the system ($K$), the gap between the achievable rate using the decentralized and centralized schemes is negligible for realistic values of $M$. In other words, even with a inaccurate lower bound of the number of users in the system $K$, the proposed decentralized scheme performs very close to the centralized one.  

\begin{figure}
%\vspace{-.1 in}
\includegraphics[width=3.2 in,height=2.2 in]{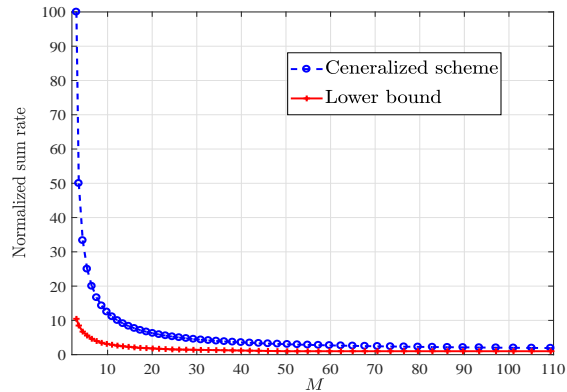}
\centering
\caption{\small The upper bound vs the lower bound for $N=K=100$.}\label{fig_uppervslower}
\end{figure}
\begin{figure}
\vspace{.1 in}
\includegraphics[width=3.2 in,height=2.2 in]{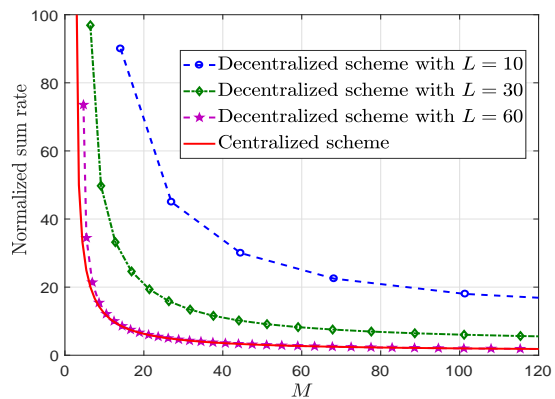}
\centering
\caption{\small Achievable rates via decentralized and centralized schemes $N=K=100$.}\label{fig_decen_vs_cen}
\end{figure}
\vspace{-.15 in}
\section{Conclusions }\label{sec:con}
In this work, we have characterized the fundamental limits of secure device-to-device coded caching systems. We have investigated a cache-aided network, where the users' requests must be served via D2D communications only. We have imposed secure caching constraint on all users, i.e., a user cannot obtain any information about any file that he had not requested. We have developed an achievable centralized coded caching scheme for this network, where the server encodes each file using a proper non-perfect secret sharing scheme and generates a set of random keys. The resulting shares and keys are carefully placed in the users' cache during the cache placement phase. After announcing the users' demands, each user transmits a one-time padded signal to the remaining users. In addition, we have provided a sequential decentralized scheme that does not require the knowledge of the number of the active users for cache placement. As a byproduct of the proposed achievability schemes, the system also keeps the files secret from any external eavesdropper that overhears the delivery phase, i.e., the secure delivery is also guaranteed. We have derived a lower (converse) bound based on cut-set arguments. Furthermore, we have shown that our proposed scheme is order-optimal and optimal in the large memory region. Our numerical results indicate that the gap between the lower and upper bounds decreases as the cache memory capacity increases. Similarly, the performance of centralized and decentralized schemes are very close for large memories. Overall, we have shown that the D2D communications can replace the server in the delivery phase with a negligible transmission overhead. This offers an affirmation of D2D communications' significant role in upcoming communication systems.     
    \appendices
    \vspace{-.1 in} 
\section{Proof of Theorem \ref{them:gap} }
First, we show that the multiplicative gap between the achievable rate in \cite{ravindrakumar2016fundamental} for multicast coded caching and the achievable rate in Theorem \ref{them:cen} can be bounded by a constant. We recall the upper and lower bounds from \cite{ravindrakumar2016fundamental}.
\begin{equation}
 R_{\text{Multicast}}\triangleq\frac{K(N+M-1)}{N+(K+1)(M-1)},  
\end{equation}    
\begin{equation}\label{multicatlower}
R_{\text{Multicast}}^{*}\geq \! \max_{s\in\{1,2,\ldots,\min(K,N/2)\}}\frac{(s\lfloor N/s\rfloor\!-\!1)\!-\!(s\!-\!1)M}{(\lfloor N/s \rfloor\!-\!1)}. 
\end{equation}

Therefore, we have 
\begin{align}
\frac{R_T^C}{R_{\text{Multicast}}}=\frac{2(N+(K+1)(M-1))}{1+(M-1)K+\sqrt{(1-(M-1)K)^2-4KN}}.
\end{align}
To simplify the notation, let $U=M-1$ and $V=\sqrt{(1-KU)^2-4KN}$. Then, we have  
%=2\frac{(N+KU+U)}{1+KU+V},
\begin{align}
\frac{R_T^C}{R_{\text{Multicast}}}\!&\!=\frac{2KU}{1\!+\!KU\!+\!V}\!+\!\frac{2U}{1\!+\!KU\!+\!V}\!+\!2\frac{N}{1\!+\!KU\!+\!V},\\
&\leq 2+1+2\frac{N}{1+KU+V}.
\end{align}
Note that the minimum value of $M$=$\frac{N}{K-1}+2$, thus we have $U\geq \frac{N}{K-1}$ and 
%&\leq 4+2\frac{N}{1+KU+V},\\
\begin{align}
\frac{R_T^C}{R_{\text{Multicast}}}\leq 3+2\frac{N}{1+\frac{KN}{K-1}+V}\leq 5 =c'.
\end{align} 
Now, consider 
\begin{align}
\frac{R_T^C}{R_T^*}&=\frac{R_T^C}{R_{\text{Multicast}}}\times \frac{R_{\text{Multicast}}}{R_T^*}\leq\frac{R_T^C}{R_{\text{Multicast}}}\times \frac{R_{\text{Multicast}}}{R_{\text{Multicast}}^*},\label{apped1}\\
&\leq c^{'} \times c^{''}=c.
\end{align}
We observe that for any value for $s$, the RHS of (\ref{d2dlower}) equals $\frac{K}{K-s}$ RHS of (\ref{multicatlower}), thus we have $R_{\text{Multicast}}^*\leq R_{T}^*$ and we can get we get (\ref{apped1}). The last step follows from \cite[Theorem 3]{ravindrakumar2016fundamental}, i.e., $\frac{R_{\text{Multicast}}}{R_{\text{Multicast}}^*}\leq c^{''}$, where $c^{''}$ is a constant independent of the system parameters. 
 \bibliographystyle{IEEEtran}
\bibliography{IEEEabrv,cachingLib}

\end{document}